\documentclass[aps,pra,groupedaddress,superscriptaddress,amsmath,amssymb,showpacs,twocolumn,footinbib,floatfix]{revtex4-1}

\usepackage{graphicx}
\usepackage{bm,bbm,braket}
\usepackage{color,booktabs,xcolor}
\usepackage{mathtools,amsthm}
\usepackage{hyperref}
\usepackage{microtype,enumerate}

\usepackage{times,txfonts}

\usepackage{multirow,array}

\usepackage{mathdots}

\usepackage{thmtools, thm-restate}

\DeclareMathOperator{\Tr}{Tr}
\DeclareMathOperator{\diag}{diag}
\let\Re\relax

\DeclareMathOperator{\Re}{Re}

\DeclareMathOperator{\cl}{cl}
\DeclareMathOperator{\co}{co}

\newcommand{\Herm}{\mathbb{H}}

\renewcommand{\H}{\mathcal{H}}
\newcommand{\F}{\mathcal{F}}
\renewcommand{\O}{\mathcal{O}}
\renewcommand{\S}{\mathcal{S}}
\renewcommand{\P}{\mathcal{P}}
\newcommand{\T}{\mathcal{T}}
\newcommand{\B}{\mathcal{B}}
\newcommand{\V}{\mathcal{V}}
\newcommand{\C}{\mathcal{C}}
\newcommand{\M}{\mathcal{M}}

\newcommand{\RR}{\mathbb{R}}

\newcommand{\OO}{\O(\bs{\lambda}\rightarrow\bs{\mu})}
\newcommand{\OOG}{\O_G(\bs{\lambda}\rightarrow\bs{\mu})}

\newcommand{\QCM}{\text{\rm QCM}}

\newcommand{\lset}{\left\{\left.}
\newcommand{\rset}{\right\}}
\renewcommand{\bar}{\;\rule{0pt}{9.5pt}\right|\;}
\newcommand{\<}{\left\langle}
\renewcommand{\>}{\right\rangle}

\let\mathds\mathbbm

\newcommand{\texteq}[1]{\stackrel{\mathclap{\mbox{\text{\scriptsize #1}}}}{=}}

\newcommand{\textgeq}[1]{\stackrel{\mathclap{\mbox{\text{\scriptsize #1}}}}{\geq}}
\newcommand{\id}{\mathds{1}}
\newcommand{\bs}[1]{\boldsymbol{#1}}

\newtheorem{thm}{Theorem}
\newtheorem*{thm*}{Theorem}
\newtheorem{prop}{Proposition}
\newtheorem*{prop*}{Proposition}
\newtheorem{lemma}{Lemma}
\newtheorem*{lemma*}{Lemma}

\newtheorem*{cor*}{Corollary}

\newtheorem*{cj*}{Conjecture}

\theoremstyle{definition}
\newtheorem{definition}{Definition}
\newtheorem*{definition*}{Definition}
\newtheorem*{rem}{Remark}

\newtheorem{innercustompost}{Postulate}
\newenvironment{custompost}[1]
  {\renewcommand\theinnercustompost{#1}\innercustompost}
  {\endinnercustompost}

\makeatletter
  \adddialect\l@ENGLISH\l@english
    \adddialect\l@en\l@english
\makeatother
\AtBeginDocument{
\heavyrulewidth=.08em
\lightrulewidth=.05em
\cmidrulewidth=.03em
\belowrulesep=.65ex
\belowbottomsep=0pt
\aboverulesep=.4ex
\abovetopsep=0pt
\cmidrulekern=.5em
}

\newcommand{\notts}{\affiliation{School of Mathematical Sciences and Centre for the Mathematics and Theoretical Physics of Quantum Non-Equilibrium Systems, University of Nottingham, University Park, Nottingham NG7 2RD, United Kingdom}}

\begin{document}

\title{Gaussian quantum resource theories}

\author{Ludovico Lami}\email{ludovico.lami@gmail.com}
\notts

\author{Bartosz Regula}\email{bartosz.regula@gmail.com}
\notts

\author{Xin Wang}
\affiliation{Centre for Quantum Software and Information, Faculty of Engineering and Information Technology, University of Technology Sydney, NSW 2007, Australia}

\author{Rosanna Nichols}
\notts

\author{Andreas Winter}
\affiliation{F\'{\i}sica Te\`{o}rica: Informaci\'{o} i Fen\`{o}mens Qu\`{a}ntics, Departament de F\'{i}sica, Universitat Aut\`{o}noma de Barcelona, ES-08193 Bellaterra (Barcelona), Spain}
\affiliation{ICREA --- Instituci\'o Catalana de Recerca i Estudis Avan\c{c}ats, Pg.~Lluis Companys 23, ES-08010 Barcelona, Spain}

\author{Gerardo Adesso}
\notts

\date{\today}

\begin{abstract}
We develop a general framework to assess capabilities and limitations of the Gaussian toolbox in continuous variable quantum information theory. Our framework allows us to characterize the structure and properties of quantum resource theories specialized to Gaussian states and Gaussian operations, establishing rigorous methods for their description and yielding a unified approach to their quantification. We show in particular that, under a few intuitive and physically motivated assumptions on the set of free states, no Gaussian quantum resource can be distilled with free Gaussian operations, even when an unlimited supply of the resource state is available. This places fundamental constraints on state manipulations in all such Gaussian resource theories.
We discuss in particular the applications to quantum entanglement, where we extend previously known results by showing that Gaussian entanglement cannot be distilled even with Gaussian operations preserving the positivity of the partial transpose, as well as to other Gaussian resources such as steering and optical nonclassicality. A comprehensive semidefinite programming representation of all these resources is explicitly provided.
\end{abstract}

\maketitle


\section{Introduction}

Continuous-variable (CV) systems of quantum harmonic oscillators play a prominent role in quantum science,  due to their ubiquitous presence, outstanding theoretical importance, and practical relevance in many quantum technologies~\cite{braunstein_2005,cerf_2007,adesso_2014}. Among them, so-called \textit{Gaussian states} are privileged as being remarkably affordable to produce and control in laboratory, while retaining, together with Gaussian operations, a significant part of the power of quantum information processing~\cite{adesso_2007-2,wang_2007,weedbrook_2012,adesso_2014}. However, such a restricted set of resources is insufficient to realize fundamental tasks like fault-tolerant quantum computation~\cite{knill_2001}, entanglement distillation~\cite{eisert_2002-1,fiurasek_2002,giedke_2002}, error correction~\cite{niset_2009}, or optimal metrology~\cite{adesso_2009}, and needs to be supplemented by nonlinear elements, e.g., photon detectors~\cite{duan_2000,takahashi_2010,knill_2001}, to achieve universality. A thorough investigation of properties and limitations of the Gaussian paradigm is thus crucial to deepen our theoretical understanding of quantum optics and information and to set suitable experimental benchmarks in practical tasks.

\begin{figure}[t]
	\centering
	\includegraphics[width=5.5cm]{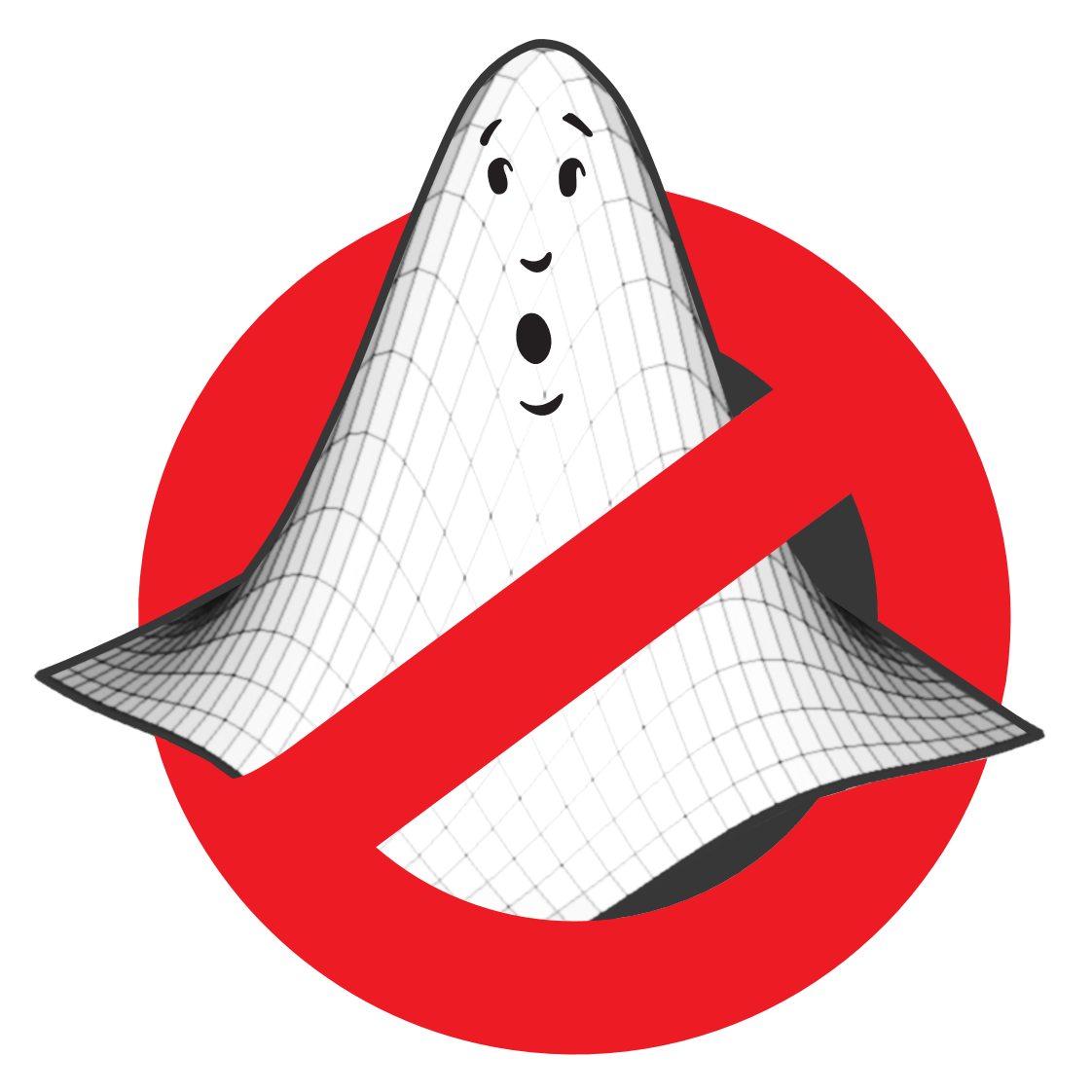}
	\caption{A general framework is developed to assess capabilities and limitations of the Gaussian toolbox in continuous variable quantum information theory. It is shown that under a few assumptions on the set of free states, no Gaussian quantum resource can be distilled with free Gaussian operations, even when an unlimited supply of the resource state is available. We refer to this general no-go result as the {\it Gaussbusters} Theorem.}
	\label{busterfig}
\end{figure}

The formalism of \textit{quantum resource theories}~\cite{horodecki_2012, sperling_2015, brandao_2015, delrio_2015, coecke_2016, gour_2017, regula_2018, anshu_2017, lami_2018} lends itself well to the investigation of such features and restrictions. Different quantum phenomena have been recently recognized and characterized as resources, including entanglement~\cite{horodecki_2009}, asymmetry~\cite{gour_2008,lostaglio_2015}, athermality~\cite{brandao_2013}, purity~\cite{horodecki_2003}, nonlocality~\cite{gallego_2012,vicente_2014}, coherence~\cite{aberg_2006,gour_2008,baumgratz_2014,streltsov_2017}, nonclassicality~\cite{vogel_2014,killoran_2016,theurer_2017,das_2017}, Einstein-Podolsky-Rosen (EPR) steering~\cite{gallego_2015}, contextuality~\cite{grudka_2014}, magic~\cite{veitch_2014,ahmadi_2017}, non-Markovianity~\cite{wakakuwa_2017}, and noiseless classical or quantum communication in quantum Shannon theory~\cite{Devetak_2005}. However, since each such resource may require a completely different approach to describe it, the alluring task of establishing a unified framework is very challenging. Although some general statements about quantum resource theories can be derived from suitable assumptions~\cite{horodecki_2012,brandao_2015,gour_2017,liu_2017,regula_2018,anshu_2017}, most research to date focused on finite-dimensional scenarios, and despite an increasing interest in developing a resource-theoretic approach to quantum optics~\cite{eisert_2003-1,idel_2016,zhang_2016,brown_2016,tan_2017,friis_2017}, there are no results that apply to a large class of resources in infinite dimensions.

In this paper we extend the general formalism of resource theories to CV systems, by introducing a framework for the study of resources whose free states (Sec.~\ref{sec:freestates}) and operations (Sec.~\ref{sec:freeoperations}) are Gaussian. This allows us to exploit the general resource-theoretic formalism to ultimately assess how powerful Gaussian states and operations are in CV quantum information theory. We show that all Gaussian resource theories which satisfy a set of physically-motivated conditions share a common structure, allowing us to simplify the description and quantification of many fundamental resources such as entanglement, EPR steering, and nonclassicality (squeezing). We establish universal constraints on state transformations under free Gaussian operations in any such resource theory, showing in Sec.~\ref{sec:nogo} that the operational task of resource distillation is impossible if one is restricted to Gaussian states and free Gaussian operations (see Fig.~\ref{busterfig}). In particular, we generalize the no-go theorem of Ref.~\cite{giedke_2002} by showing that Gaussian entanglement cannot be distilled by Gaussian operations preserving the positivity of partial transpose  --- a larger class than previously considered --- and we prove equivalent no-go results for other relevant Gaussian resources. Detailed examples and applications are illustrated in Sec.~\ref{sec examples}.
We discuss our main results below, deferring technical derivations to the Appendix.


\section{Free states}\label{sec:freestates}

Let us briefly recall the basics of Gaussian states~\cite{weedbrook_2012,adesso_2014,serafini_2017}. Mathematically, a $n$-mode CV system is identified by a collection of canonical operators $x_1,p_1,\ldots, x_n, p_n$, which we can arrange in a vector $r\coloneqq (x_1,\ldots, x_n, p_1, \ldots, p_n)^T$. The canonical commutation relations $[x_j,p_k]=i\delta_{jk}$ can then be cast as $[r,r^T]= i\Omega$, where $\Omega \coloneqq \left(\begin{smallmatrix} 0 & \mathds{1} \\ -\mathds{1} & 0 \end{smallmatrix} \right)$ is the symplectic form. Denoting by $\mathcal{G}_n$ the set of $n$-mode Gaussian states, any $\rho \coloneqq \rho_G[V,s] \in \mathcal{G}_n$ is fully specified by its (real) displacement vector $s \coloneqq \Tr [r\, \rho]$ and its (real, symmetric) covariance matrix  $V \coloneqq \Tr [ \{r - s, r^T - s^T\}\,\rho]$ with $\{\cdot, \cdot\}$ being the anticommutator. Letting $\M_{2n}(\RR)$ denote the set of all real $2n \times 2n$ matrices, we will call
\begin{equation}\QCM_n \coloneqq \lset V \in \M_{2n}(\RR) \bar V = V^T,\, V \geq i \Omega \rset\end{equation}
the set of quantum covariance matrices, i.e., those $V$ that satisfy the Robertson-Schr\"odinger uncertainty principle~\cite{simon_1994}. Note that any such $V \in \QCM_n$ is strictly positive definite~\cite{serafini_2017}.

Resource theories are built upon two main notions~\cite{brandao_2015}: (i) the subset $\F$ of free states, i.e., those which do not possess the given resource; and (ii) the subset $\O$ of free operations, i.e., those quantum channels unable to generate the resource, specified by the physical constraints of the theory. As free states can be prepared by free operations at no cost, during a protocol one may add ancillary systems to one's original system; following~\cite{brandao_2015}, we will refer to such systems as spatially separated.

In any resource theory, there may be different ways to define the set of free states (think e.g.~of entanglement theory, in which one needs to specify a partition to identify separable states). To address this, we assume that each of the spatially separated subsystems $j=1,\ldots, l$ is fully specified by a set $\lambda_j$ of variables, which can then be grouped in a single vector $\bs{\lambda}\coloneqq (\lambda_1, \ldots, \lambda_l)$. For instance, one such variable will be the total number of modes $n_j$ of each subsystem $j$. For a CV system made of $l$ spatially separated subsystems identified by a vector of variables $\bs{\lambda}$, we then denote by $\F(\bs{\lambda})$ the subset of free states, and by $\F\coloneqq \bigcup_{\bs{\lambda}} \F(\bs{\lambda})$ the set of all free states over arbitrary collections of spatially separated subsystems. 

Since we care about the Gaussian restriction of any resource theory, we will focus on the set $\F_G\coloneqq \bigcup_{\bs{\lambda}} \F_G(\bs{\lambda})$ of free \emph{Gaussian} states, where  $\F_G(\bs{\lambda})\coloneqq \F(\bs{\lambda})\cap \mathcal{G}_N$, with $N=\sum_j n_j$, for a fixed $\bs{\lambda}$, and the corresponding set of free covariance matrices is
\begin{equation}\begin{aligned}
	\V_\F(\bs{\lambda}) \coloneqq \lset V\in \mathcal{M}_{2N}(\RR) \bar \exists\, s\in\mathds{R}^{2N}:\ \rho_G[V,s]\in \F(\bs{\lambda}) \rset.
\end{aligned}\end{equation}

There are some standard assumptions about the set of free states, formalized as Postulates~\ref{post tensor}-\ref{post weak closed} in~\cite{brandao_2015}, that have a sound theoretical basis and apply to a wide range of theories. We will therefore regard them as a safe starting point to establish a set of fundamental requirements for Gaussian resource theories. Before proceeding further, let us make a first working assumption that will simplify the following analysis considerably:
\begin{custompost}{0} \label{post D invariance}
The set of free states is invariant under displacement operations.
\end{custompost}
To justify this assumption, note that displacement operations can be applied to any system by adding an ancillary system in a highly excited coherent state, and combining the two systems at a low-transmissivity beam splitter~\cite{paris_1996}. From an experimental standpoint, coherent states and beam splitters are relatively cheap tools.
Crucially, Postulate~\ref{post D invariance} implies that the set of free Gaussian states is now fully described by the corresponding covariance matrices, so we can write $\V_\F (\bs{\lambda}) = \lset V\in \mathcal{M}_{2N}(\RR) \bar \rho_G[V,0]\in \F(\bs{\lambda})\rset$ and $\F_G(\bs{\lambda}) = \lset \rho_G[V,s] \bar V\in \V_\F(\bs{\lambda}),\ s\in\mathds{R}^{2N} \rset$.
The following assumptions define the structure of the theory for composite systems.
\begin{custompost}{I}\label{post tensor} The set of free states is closed under tensor products of spatially separated subsystems.\end{custompost}\vspace*{-\baselineskip}
\begin{custompost}{II}\label{post partial trace} The set of free states is closed under partial traces over spatially separated subsystems.\end{custompost}\vspace*{-\baselineskip}
\begin{custompost}{III}\label{post permutation} The set of free states is closed under permutations of spatially separated subsystems.\end{custompost}
These three properties carry over to the restricted set of free Gaussian states $\F_G$, since it is well known that Gaussian states are also closed under the above operations. Postulate~\ref{post tensor} can be rewritten symbolically as $\F(\bs{\lambda}) \otimes \F(\bs{\lambda'}) \subseteq \F(\bs{\lambda} \oplus \bs{\lambda'})$, where for $\bs{\lambda}=(\lambda_1,\ldots, \lambda_l)$ and $\bs{\lambda'}=(\lambda_1, \ldots, \lambda_{l'})$ one sets $\bs{\lambda}\oplus \bs{\lambda'} \coloneqq (\lambda_1, \ldots, \lambda_l, \lambda'_1, \ldots, \lambda'_{l'})$.  At the level of covariance matrices, this translates to $\V_\F(\bs{\lambda}) \oplus \V_\F(\bs{\lambda'}) \subseteq \V_\F(\bs{\lambda}\oplus \bs{\lambda'})$. Similarly, we can formulate Postulate~\ref{post partial trace} as $\Pi\, \V_\F(\bs{\lambda}\oplus \bs{\lambda'})\, \Pi^T \subseteq \V_\F(\bs{\lambda})$, where $\Pi$ is the projector onto the subsystems corresponding to $\bs{\lambda}$.

Apart from the above Postulates, one of the the most basic assumptions that we can make about the unrestricted set (i.e., before the intersection with $\mathcal{G}_N$) of free states is undoubtedly \emph{convexity}. This is justifiable as in most physically relevant cases we do not expect an increase in quantum resources under classical mixing \footnote{Note, however, that one may define also non-convex resources --- notably, we obtain one such resource theory by letting the set of all Gaussian states to be free --- but we will not study them here.}. At the level of all states, we have then:
\begin{custompost}{IV} \label{post convexity}
For all $\bs{\lambda}$, the set of free states $\F(\bs{\lambda})$ is convex.
\end{custompost}
Our final Postulate will pertain to the closedness of the set of free states in the Banach space $\T(\H_N)$ of trace-class operators acting on the Hilbert space $\H_N$, endowed with the trace norm $\|\cdot\|_1$.

\begin{custompost}{V} \label{post weak closed}
For all $\bs{\lambda}$, the set of free states $\F(\bs{\lambda}) \subseteq \T(\H_N)$ is norm-closed.
\end{custompost}

Since infinite-dimensional spaces admit many legitimate linear topologies with respect to which we can define closedness, the above choice of the norm topology may seem rather arbitrary. However, it turns out that in the present context any other reasonable choice still yields the same result. In fact, Mackey's theorem~\cite[Thm.~8.9]{aliprantis_2007} ensures that all linear topologies on $\T(\H_N)$ that agree on the set of continuous linear functionals possess the same closed convex sets. To see why this is physically relevant, remember that the norm-continuous linear functionals on $\T(\H_N)$ can be written as $\rho\mapsto \Tr[\rho A]$, where $A\in \B(\H_N)$ is a generic bounded operator, i.e., an observable. Summarizing the above discussion: Postulates~\ref{post convexity} and~\ref{post weak closed} together imply that all sets of free states $\F(\bs{\lambda})$ and $\F_G(\bs{\lambda})$ are closed with respect to any linear topology whose corresponding continuous linear functionals are (all) the observables.
For further details on the issue of closedness of the relevant sets, refer to Appendix~\ref{app topoogy}.

Let us now discuss some consequences of the above assumptions in the Gaussian setting. In order to do so, it is important to understand the topology of the set of Gaussian states in some detail. In Appendix~\ref{app topoogy}, we show that Gaussian states form a closed set with respect to the trace norm topology (Lemma~\ref{lemma G closed}), and that the map $(V,s)\mapsto \rho_G[V,s]$ that sends a pair formed by a quantum covariance matrix and a real vector to the corresponding Gaussian state is continuous with respect to the same topology (Lemma~\ref{lemma V to rho 1-continuous}). A key difference between a Gaussian resource theory satisfying Postulates~\ref{post tensor}-\ref{post weak closed} and a corresponding finite-dimensional theory is that the set of Gaussian states is non-convex, hence $\F_G(\bs{\lambda})$ can not be expected to be convex either. We have instead a weaker property that we dub \textit{Gaussian convexity}: if a trace norm limit of convex combinations of free Gaussian states is a Gaussian state, then it must be free. Importantly, this implies the \textit{upward closedness} of the set of free covariance matrices~$\V_\F (\bs{\lambda})$, formalized as follows.
\begin{prop} \label{prop V}
When Postulates~\ref{post convexity} and~\ref{post weak closed} hold, the set $\V_\F(\bs{\lambda})$ is topologically closed as well as upward closed, in the sense that, if $V\in \V_\F(\bs{\lambda})$ and $W\geq V$, then $W\in\V_\F(\bs{\lambda})$.
\end{prop}
Another desirable property of the set of free covariance matrices $\V_\F(\bs{\lambda})$ is for itself to be convex. Interestingly, this does not follow directly from our Postulates, but is indeed implied by an additional natural assumption, i.e., that the given set of free Gaussian states $\F_G(\bs{\lambda})$ is closed under mode-by-mode mixing with 50:50 beam splitters (Prop.~\ref{prop convexity}).


\section{Free operations and quantification}\label{sec:freeoperations}

In any resource theory, a free operation can be any channel which always maps free states into free states. However, the physical setting of the given resource can further restrict the allowed free operations: for instance, in entanglement theory, the distant laboratories paradigm leads to the set of local operations and classical communication (LOCC). To keep our results as general as possible, we will consider the \textit{maximal} set of resource non-generating operations, and only impose the natural restriction that
free operations should also be Gaussian, i.e., such that they always map a Gaussian state to a Gaussian state~\cite{giedke_2002, depalma_2015}.
\begin{definition}
A quantum channel $\Lambda:\T(\H_N)\rightarrow \T(\H_M)$ is called \textit{resource non-generating} if $\Lambda \left[ \F(\bs{\lambda}) \right] \subseteq \F(\bs{\mu})$ for systems described by variables $\bs{\lambda},\bs{\mu}$. The set of all resource non-generating operations is denoted by $\O(\bs{\lambda}\rightarrow\bs{\mu})$, and the restriction to Gaussian operations by $\O_G(\bs{\lambda}\rightarrow\bs{\mu})$. In particular, $\Lambda \left[ \F_G(\bs{\lambda}) \right] \subseteq \F_G(\bs{\mu})$ for all $\Lambda \in \O_G(\bs{\lambda}\rightarrow\bs{\mu})$.
\end{definition}

A fundamental question in any resource theory is, given several resourceful states, to quantify the degree of their resourcefulness and thus compare the usefulness of the states in operational tasks~\cite{plenio_2007,brandao_2015,regula_2018,anshu_2017}. For this, one needs a measure $\mu : \T(\H_N) \rightarrow \RR_+$ which satisfies two basic criteria: faithfulness, i.e., being minimum on all  (and only on) free states,  $\mu (\rho) = \inf_{\sigma\in\T(\H_{N})} \mu(\sigma) \iff \rho \in \F(\bs\lambda)$, as well as monotonicity, i.e., $\mu(\Lambda(\rho)) \leq \mu(\rho)$ for all free operations $\Lambda$. Here we stress that we can consider the maximal set of free operations $\OO$ without loss of generality, since any measure monotonic under $\OO$ will also be a monotone under any smaller subset of free operations. Analogously, in the setting of Gaussian resources, we will be interested in quantifiers $\mu_G : \QCM_N \rightarrow \RR_+$ defined at the level of covariance matrices and monotonic under the free Gaussian operations, so that $\mu_G(V') \leq \mu_G(V)$ where $V'$ is the covariance matrix corresponding to the state $\Lambda(\rho_G[V,s])$ with $\Lambda \in \O_G(\bs{\lambda}\rightarrow\bs{\mu})$.

A general instance of such a measure --- a variant of which has been considered in the characterization of entanglement before~\cite{giedke_2002} --- can be defined for any $V\in \QCM_N$ as
\begin{equation}
\kappa_\F (V) \coloneqq \min\lset t \geq 1 \bar t V\in \V_\F(\bs{\lambda}) \rset.
\label{kappa}
\end{equation}
The measure can be easily seen to be faithful in the sense that $\kappa_\F(V) \geq 1$ and $\kappa_\F(V) = 1$ iff $V \in \V_\F(\bs\lambda)$, and the fact that the set on the right-hand side of Eq.~\eqref{kappa} is non-empty is ensured by the upward closedness of $\V_\F(\bs\lambda)$. The properties and monotonicity of the above quantifier can be summarized as follows.
\begin{prop}\label{prop:kappa}
The function $\kappa_\F(\cdot)$ is finite and well-defined on all covariance matrices, faithful, continuous, and monotonic under $\OOG$.
\end{prop}
We defer the proof to Appendix~\ref{app G resources} (see Prop.~\ref{prop:kappa-reformul}).
Note that, if membership of the set $\V_\F(\bs{\lambda})$ can be decided by semidefinite constraints at the level of covariance matrices, the evaluation of $\kappa_\F(V)$ reduces to a semidefinite program. We discuss such cases in Section~\ref{sec examples} and in Appendix~\ref{app SDP}.


\section{No-go theorem for Gaussian resource distillation}\label{sec:nogo}

At the heart of every resource theory lies the problem of characterizing state transformations which are allowed by the given set of free operations. In particular, the operational task of {\it resource distillation} deals with using free operations to convert multiple copies of a given quantum state into a smaller number of target states, usually representing maximally resourceful states. This task was first considered in the resource theory of entanglement with LOCC~\cite{bennett_1996-4,bennett_1996-5}, and has been later extended to more general settings~\cite{rains_1999,rains_2001,brandao_2011,buscemi_2010-2,fang_2017} and other quantum resources~\cite{bravyi_2005,winter_2016,regula_2017-2}. Entanglement distillation has also been considered for Gaussian states~\cite{giedke_2002,fiurasek_2002,eisert_2002-1}, where the task can be expressed as using LOCC to transform multiple copies of a bipartite state $\rho_{AB}^{\otimes n}$ into a state which approaches a maximally entangled state as $n \to \infty$.
An archetypal example of the analysis of the limitations of the Gaussian paradigm in this context has been carried out in~\cite{giedke_2002}, where it was shown that Gaussian LOCC protocols are not sufficient to distill Gaussian entanglement.

Since the existence of a ``golden unit'' or a unique maximally resourceful state is not guaranteed in arbitrary quantum resource theories, we can consider the more general task of approximately converting multiple copies of a quantum state into another state which is more resourceful; that is, given a Gaussian state with covariance matrix $V$, we ask about the existence of free operations that implement the transformations $V^{\oplus n} \rightarrow W_{n}$, for some sequence of covariance matrices $W_{n}$ that approach a fixed target $W$ such that $\kappa_\F (W) > \kappa_\F (V)$.
A central result of this work is a general {\it no-go} result entailing that, in any resource theory in our framework, the distillation of the given resource with free Gaussian operations is {\it de facto} impossible, as illustrated in Fig.~\ref{busterfig}. This result is in stark contrast with the main finding of~\cite{brandao_2015}, which instead implies the complete reversibility of the considered resource theory. Such a dramatic difference in the conclusions is even more surprising when one considers that the starting postulates are quite similar in the two cases, and illustrates clearly the intrinsic limitations of the Gaussian framework.
\begin{thm}[{\bf Gaussbusters}] \label{thm:no-go}
Consider an arbitrary Gaussian resource theory satisfying Postulates~\ref{post D invariance}-\ref{post weak closed} and two covariance matrices $V,W\in \text{\emph{QCM}}_N$. If $\kappa_\F (W) > \kappa_\F (V)$, then it is impossible to find a sequence $(W_{n})_{n\in\mathds{N}}\subset\text{\emph{QCM}}_{N}$ such that $\lim_{n\rightarrow\infty} W_{n}=W$ and the transformations $V^{\oplus n}\rightarrow W_{n}$ are possible with Gaussian resource non-generating operations for all $n$.
\end{thm}
The proof of the above theorem relies on a special property of the measure $\kappa_\F$ that we could call, borrowing terminology from classical probability theory, \emph{tensorization property}~\cite{beigi_2016}: for all resource theories in consideration, $\kappa_\F$ does not change when multiple copies of a quantum state are considered; more generally, we have $\kappa_\F \left(V \oplus W\right) = \max\{ \kappa_\F(V), \kappa_\F (W)\}$ for any two covariance matrices $V,W$ (Lemma~\ref{kappa multi-copy}). This, together with the monotonicity of $\kappa_\F$, immediately implies that distillation is impossible since we cannot increase $\kappa_\F$ with free Gaussian operations.
In the following, we present explicit applications of our framework to a broad set of continuous variable resources, namely squeezing (equivalently, nonclassicality), quantum entanglement manipulated via local operations and classical communication or via operations preserving the positivity of the partial transpose, and steering.

\section{Examples and applications} \label{sec examples}


Quite remarkably, it turns out that in many --- if not all --- physically relevant resource theories, $\V_\F(\bs{\lambda})$ is not only a convex set, but can even be described by means of semidefinite programming (SDP) constraints. Although we leave open the question of whether a general principle can be found from which the existence of such a description follows naturally, we will now characterize the quantification of all resources for which such SDP structure is known to exist. In particular, our results apply to any resource theory satisfying Postulates~\ref{post D invariance}--\ref{post weak closed} whose set of free states can be described by constraints of the kind
\begin{equation}\begin{aligned}
	\V_\F(\bs{\lambda}) = \lset V \in \QCM_{N} \bar V \geq f(Q) + C,\ g(Q)\geq D \rset
\end{aligned} \label{V as SDP}
\end{equation}
where $Q$ is a Hermitian matrix variable of some fixed size, $f$ and $g$ are linear functions, and $C$ and $D$ are constant Hermitian matrices. The main advantage of the representation in Eq.~\eqref{V as SDP} is that the associated quantifier $\kappa_\F$ in Eq.~\eqref{kappa} can then be evaluated via an efficiently computable semidefinite program:
\begin{equation}\label{kappa SDP}
\begin{aligned}
\hspace{-5em}\kappa_\F (V) = \ & \underset{\xi,\, Q}{\text{minimize}}
& & \xi \\
& \text{subject to} & & \xi \, V \geq f(Q) + C\\
&&& g(Q)\geq D \\
&&& \xi \geq 1.
\end{aligned}
\end{equation}

Alternatively, one can choose to introduce the quantity
\begin{equation}\begin{aligned}
	\upsilon_\F(V) \coloneqq \max_{\zeta,\, Q} \lset \zeta \bar V \geq \zeta (f(Q) + C),\ g(Q)\geq D  \rset
\end{aligned}\end{equation}
in which case $\kappa_\F(V) = \max \left\{ 1,\, 1/\upsilon_\F(V)\right\}$. The advantage this formulation is that the dual of the optimization problem $\upsilon_\F$ can be expressed by means of the so-called resource witnesses based on second moments~\cite{hyllus_2006-1}, that is, as an optimization over the expectation values $\Tr (W V)$ at the level of the covariance matrix. Assuming that strong duality for the problem \eqref{kappa SDP} holds (which can be straightforwardly verified for all of the considered resource theories), we then have the SDP
\begin{equation}
\begin{aligned}
\hspace{-5em}\upsilon_\F (V) = \ & \underset{W,\, Y}{\text{minimize}}
& & \Tr( W V )\\
& \text{subject to} & & \Tr (W C) + \Tr (Y D) = 1\\
&&& f^\dagger (W) = g^\dagger(Y) \\
&&& W,Y \geq 0
\end{aligned}
\end{equation}
where $f^\dagger, g^\dagger$ are the adjoint maps, that is, the unique linear maps satisfying $\Tr(f(A)B)=\Tr(A f^\dagger(B))$ for any Hermitian $A,B$.

Many common Gaussian resources can indeed be expressed and quantified in this way --- we will now provide some representative examples of such resources, which we have also collected in Table~\ref{tab:resources}.

\begin{table*}
\begin{ruledtabular}
\begin{tabular}{llllll}
Resource & $f(Q)$ & $g(Q)$ & $C$ & $D$ & Further constraints on $Q$\\\midrule
Bipartite entanglement~\cite{werner_2001} & $Q_A \oplus Q_B$ & $Q_A \oplus Q_B$ & --- & $i \Omega_{A} \oplus i\Omega_{B}$ & $Q_A = Q_A^T$, $Q_B = Q_B^T$\\
Bipartite entanglement (simplified)~\cite{lami_2016} & $Q_A \oplus 0_B$ & $Q_A$ & $0_A \oplus i \Omega_B$ & $i \Omega_A$ & $Q_A = Q_A^T$\\
Negative partial transpose~\cite{werner_2001} & --- & --- & $i\Omega_A \oplus (-i\Omega_B)$  & --- &\\
Steerability ($A\to B$)~\cite{wiseman_2007} & --- & --- & $0_A \oplus i\Omega_B$ & --- &\\
Nonclassicality~\cite{simon_1994} & --- & --- & $\mathbbm{1}$ & --- &
\end{tabular}
\end{ruledtabular}
\caption{Examples of Gaussian resources whose quantification can be represented in the considered framework. The table relates the resources with the notation of Eqs.~\eqref{V as SDP},~\eqref{kappa SDP}. For the sake of clarity of presentation, some additional constraints which one has to impose on the matrix $Q$ have been considered separately in the rightmost column, although they can be explicitly brought into the form of \eqref{V as SDP} by including them in the function $g$.}
\label{tab:resources}
\end{table*}

\subsection{Squeezing (nonclassicality)}
We start by looking at the simplest Gaussian resource theory of all, namely that of squeezing or nonclassicality~\cite{lee_1991, simon_1994, serafini_2005, idel_2016}. The free states of this theory, also called \emph{classical} states from now on, are simply convex mixtures of coherent states. Within this framework, the goal is usually that of preparing squeezed states, which may be useful for some practical (e.g.~metrological) tasks~\cite{SchnabelPhysRep}. It is not difficult to see that the continuous variable resource theory of squeezing obeys all the Postulates we presented. The free operations include in particular passive transformations, obtained by concatenating~\cite{yadin_2018}: (i) the addition of ancillae in classical Gaussian states; (ii) passive unitaries, defined as those symplectic unitaries that preserve the total photon number; and (iii) destructive Gaussian measurements.
These operations are relatively cheap to realize experimentally, as passive unitaries can always be implemented by combining beam splitters and phase shifters~\cite{reck_1994}.

Restricting to the Gaussian setting, free states in this theory admit a remarkably simple description in terms of their covariance matrices, for $\rho_G[V,s]$ is a classical state if and only if $V\geq \id$~\cite{simon_1994}. This gives us the simple form
\begin{equation}
\begin{aligned}
\hspace{-5em}\kappa_\C(V) = \,& \underset{\xi\, \geq\, 1}{\text{minimize}}
& & \xi \\
& \text{subject to} & & \xi \, V \geq  \id
\end{aligned}
\end{equation}
which can be easily seen to be exactly computable as $\kappa_\C(V) = \max \left\{ 1,\, 1/\lambda_{\min}(V) \right\}$, where $\lambda_{\min}$ denotes the minimal eigenvalue. Our main result in Thm.~\ref{thm:no-go} then establishes a no-go result about the convertibility of nonclassical Gaussian states under all operations preserving the set of classical states, and in particular passive operations.

\subsection{Entanglement}
The resource theory of quantum entanglement is another example of a theory for which all of the Postulates hold~\cite{brandao_2015, adesso_2014}. Focusing on bipartite entanglement between parties $A$ and $B$ for simplicity, the set of free states is formed by the separable states $\S(A|B)$. The most operationally relevant set of free operations includes all transformations that are implementable as local operations assisted by classical communication (LOCC), and is a strict subset of all the resource non-generating (separability-preserving) operations.

In our Gaussian setting, the set $\V_\S$ can be described by semidefinite constraints of the form
\begin{equation}\begin{aligned}
	\V_\S(A|B) = \lset V \in \QCM_{N_{AB}} \bar V \geq \gamma_A \oplus \gamma_B,\ \gamma_{i} \in \QCM_{N_{i}} \rset.
\end{aligned}
\end{equation}
which can be easily expressed in the form of Eq.~\eqref{V as SDP} (see Table~\ref{tab:resources}). The associated measure $\kappa_\S$ can then be computed as a semidefinite program~\cite{hyllus_2006-1}, and corresponds to the inverse of a quantifier studied in~\cite{giedke_2002}.

Notice that Thm.~\ref{thm:no-go} includes as a particular case the result of~\cite{giedke_2002}, showing the impossibility of entanglement distillation with Gaussian LOCC: in fact, it readily generalizes the result by showing that distillation with Gaussian separability-preserving operations is also impossible.

We can strengthen the result even further by relating the resource theory of entanglement to the one of negative partial transpose, in which the free states $\P(A|B)$ are those with positive partial transpose across the cut $A|B$. This set can also be obtained from Eq.~\eqref{V as SDP} as
\begin{equation}\begin{aligned}
	\V_\P(A|B) = \lset V \in \QCM_{N_{AB}} \bar V \geq i\Omega_A \oplus (-i\Omega_B) \rset.
\end{aligned}
\end{equation}
Here, the quantifier $\kappa_\P$ admits an analytical characterization as $\kappa_\P(V_{AB}) = \max\{ 1, 1/\nu_{\min} (\widetilde{V}_{AB})\}$ with $\nu_{\min}(\widetilde{V}_{AB})$ being the smallest symplectic eigenvalue of the partially transposed covariance matrix~\cite{lami_2016}. We then notice that, for any sequence of states $\rho(n)$ which approaches the maximally entangled state in the limit $n \to \infty$, we have $\lim_{n\to\infty} \kappa_{\S}(\rho(n)) =\infty = \lim_{n\to\infty} \kappa_{\P}(\rho(n))$, and therefore the distillation of entanglement would necessarily involve increasing $\kappa_{\P}$. By Thm.~\ref{thm:no-go}, we get that Gaussian entanglement distillation is impossible even with Gaussian operations preserving the positivity of the partial transpose. Among those operations --- which can be strictly more powerful than LOCC alone --- there are for instance those transformations implementable by means of Gaussian LOCC assisted by an \emph{unlimited} supply of bound entangled Gaussian states~\cite{werner_2001}.

We further remark that the characterization of the set of separable Gaussian states and their corresponding covariance matrices can be simplified to~\cite{lami_2016}
\begin{equation}\begin{aligned}
	\V_\S(A|B) = \lset V \in \QCM_{N_{AB}} \bar V \geq \gamma_A \oplus i\Omega_B,\ \gamma_{A} \in \QCM_{N_{A}} \rset
\end{aligned}
\end{equation}
which in particular means that the computation of the quantifier $\kappa_\S$ can be performed by optimizing only over one of the subsystems --- this has particular implications for the case where one of the subsystems has a larger dimension than the other, simplifying the computation of the relevant quantities. For completeness, we give the full forms of the measures $\kappa_\S$ and $\upsilon_\S$ simplified in this way in Appendix~\ref{app SDP}.

\subsection{Steering}
Another fundamental resource theory is based on the phenomenon of EPR steering~\cite{schrodinger_1935,reid_1989}, in which party $A$ can exploit quantum correlations to influence the state of another party $B$ by only performing measurements on $A$'s subsystem. In resource-theoretic approaches to steering~\cite{wiseman_2007,gallego_2015,piani_2015}, the free states are referred to as $A\!\to\! B$ unsteerable, and free operations are commonly chosen to be one-way LOCC, reflecting the asymmetric nature of steering.
Steering admits a simplified characterization when restricted to Gaussian measurements, allowing for a dedicated resource theory of Gaussian steering to be established~\cite{wiseman_2007,kogias_2015,ji_2015,lami_2016-1,xiang_2017}. It turns out that the set of free states $\T_{A \to B}$ that are unsteerable by Gaussian measurements on $A$ can be described as~\cite{wiseman_2007}
\begin{equation}\begin{aligned}
	\V_\T(\bs{\lambda}) = \lset V \in \QCM_{N_{AB}} \bar V \geq 0_A\oplus i\Omega_B \rset.
\end{aligned}
\end{equation}
It is then easy to verify that our Postulates are satisfied, and the no-go result of Thm.~\ref{thm:no-go} holds also for the Gaussian resource theory of steering --- that is, the distillation of steering from Gaussian states is impossible by Gaussian steering non-generating operations, with the latter including all relevant classes of free operations such as one-way Gaussian LOCC. We remark that in this case the quantifier $\kappa_\T$ can be computed as
\begin{equation}
\begin{aligned}
\hspace{-5em}\kappa_\T(V_{AB}) = \,& \underset{\lambda \geq 1}{\text{minimize}}
& & \lambda\\
& \text{subject to} & & \lambda \, V_{AB} / V_A \geq i \Omega_B,
\end{aligned}
\end{equation}
with $V_{AB}/V_A$ denoting the Schur complement, which admits an analytical characterization as $\kappa_\T(V_{AB}) = \max \{ 1, \, 1/\nu_{\min}(V_{AB}/V_A) \}$. In the particular case when system $B$ consists of only one mode, $\log \kappa_\T$ is equal to a previously introduced quantifier of Gaussian steering~\cite{kogias_2015}.

\section{Conclusions} \label{sec conclusions}

We have introduced a framework for the characterization of general CV Gaussian quantum resource theories satisfying a set of intuitive constraints on their set of free states. The approach allowed us to describe many important resources such as entanglement, steering, and nonclassicality together in a common formalism, obtaining novel results in the characterization of the resources as well as shedding light onto their properties. In particular, we showed that the task of resource distillation is impossible with free Gaussian operations in the given resource theories, by proving specifically that, by such operations, one cannot convert (even infinitely many copies of) a Gaussian state into another Gaussian state with a higher resource content as quantified by the resource monotone defined in this paper. This establishes fundamental limitations of the Gaussian paradigm for state transformations. 

An interesting open question is whether some sort of converse of Thm.~\ref{thm:no-go} holds. Namely, given any Gaussian resource theory and two covariance matrices $V,W$ such that $\kappa_\F (V) \geq \kappa_\F (W)$, is it always possible to
convert a large number of copies of $V$ into a single copy of $W$ with Gaussian resource non-generating operations? Even more ambitiously, can the transformation $V\rightarrow W$ happen with asymptotic non-zero rate, if one allows for vanishing errors? These questions will be explored in further work.

In summary, our results are a step forward in the characterization of general quantum resources, bridging the gap between the different approaches to finite- and infinite-dimensional settings, and elucidating the power of Gaussian states and operations in quantum information processing.
Our work opens an avenue for further investigation of many aspects of CV resources, including a complete characterization of state transformations as well as operational tasks and protocols such as resource distillation and dilution beyond Gaussianity.

\vspace{1ex}
\emph{Acknowledgements.} We are grateful to Giacomo De Palma, Vittorio Giovannetti and Krishna Kumar Sabapathy for useful discussions on this and related topics. We acknowledge financial support from the European Research Council (ERC) under the Starting Grant GQCOP (Grant No.~637352) and from the Foundational Questions Institute under the Physics of the Observer Programme (Grant No.~FQXi-RFP-1601). AW was supported by the ERC Advanced Grant IRQUAT, and the Spanish MINECO, project FIS2016-86681-P. LL and BR contributed equally to this work.

\bibliography{main}


\appendix
\setcounter{thm}{0}
\renewcommand{\theequation}{A\arabic{equation}}
\renewcommand{\therem}{A\arabic{rem}}
\renewcommand{\thedefinition}{A\arabic{definition}}
\renewcommand{\theex}{A\arabic{ex}}
\renewcommand{\thethm}{A\arabic{thm}}
\renewcommand{\thecor}{A\arabic{cor}}
\renewcommand{\thelemma}{A\arabic{lemma}}
\renewcommand{\theprop}{A\arabic{prop}}
\renewcommand{\thequestion}{A\arabic{question}}
\renewcommand{\thecj}{A\arabic{cj}}

\section{Topology of Gaussian states} \label{app topoogy}

\subsection{Notation and definitions}

For completeness, we recall the relevant definitions and concepts. Consider a continuous variable system of $n$ modes, for which we adopt the so-called real notation. In what follows, we reserve the letter $r$ for the column vector formed by the $n$ pairs of canonically conjugated field operators, sorted as
\begin{equation}
r\coloneqq (x_{1}, \ldots, x_{n}, p_{1}, \ldots, p_{n})^{T}\, .
\label{r}
\end{equation}
Here, the transposition sign refers only to the phase space degrees of freedom, and does not act on the Hilbert space.
With the help of this notation, the \emph{canonical commutation relations} $[x_j, p_k] = i \delta_{jk}$\ can be rewritten in a compact vector form as
\begin{equation}
    [r,r^T] = i \Omega \coloneqq i \begin{pmatrix} 0 & \id \\ - \id & 0 \end{pmatrix}\, .
    \label{Omega}
\end{equation}
The displacement operator associated with $\xi\in \mathds{R}^{2n}$ is given by $D(\xi) \coloneqq e^{i \xi^{T} \Omega r}$ and satisfies the identity
\begin{equation}
D(\xi_{1}) D(\xi_{2}) = e^{-\frac{i}{2} \xi_{1}^{T}\Omega \xi_{2}} D(\xi_{1}+\xi_{2})\, ,
\label{Weyl}
\end{equation}
referred to as the \emph{Weyl form of the canonical commutation relations.} Observe that $D(\xi)^\dag=D(-\xi)$ for all real vectors $\xi$.

The displacement operators can be used to generate the notable set of coherent states. For $u\in\mathds{R}^{2n}$, one defines
\begin{equation}
    \ket{u} \coloneqq D(u) \ket{0}\, ,
    \label{coherent}
\end{equation}
where $\ket{0}$ is the vacuum state. Applying the Campbell-Baker-Hausdorff formula to the exponential that defines the displacement operator, it is not too difficult to show that
\begin{equation}
    \braket{0|u} = \braket{0|D(u)|0} = e^{-\frac14 u^T u}\, .
    \label{D on 0}
\end{equation}

Coherent states are just particular examples of Gaussian states, defined as thermal states of quadratic Hamiltonians. We denote the set of Gaussian states of an $n$-mode system by $\mathcal{G}_n$. Remember that Gaussian states can be uniquely identified by their first and second moments, respectively given by
\begin{align}
    s &\coloneqq \Tr [\rho\, r]\, \label{first} \\
   V_{jk} &\coloneqq \Tr \left[ \rho \left\{(r-s)_j, (r-s)_k \right\} \right] . \label{second}
\end{align}
Here, the anticommutator $\{H,K\}\coloneqq HK+KH$ is needed in order to make the above expression real. While any vector $s\in \mathds{R}^{2n}$ can represent the first moments of an $n$-mode Gaussian state, it is well-known that the entries of a real symmetric matrix $V$ are the second moments of some Gaussian state if and only if
\begin{equation}
V\geq i\Omega\, ,
\label{Heisenberg}
\end{equation}
the above relation encoding the constraints coming from Heisenberg's uncertainty principle in this context. Real symmetric matrices satisfying Eq.~\eqref{Heisenberg} are called quantum covariance matrices in what follows. It can be shown that every such matrix is necessarily strictly positive, i.e., Eq.~\eqref{Heisenberg} implies that $V>0$.

For every trace class operator $T$, it is convenient to define its \emph{characteristic function}
\begin{equation}
\chi_{T}(\xi) \coloneqq \Tr[T D(\xi)]\, .
\label{chi}
\end{equation}
The operator can be reconstructed from its characteristic functions by means of the following relation~\cite[Cor. 5.3.5]{holevo_2011}:
\begin{equation}
T = \int \frac{d^{2n}\xi}{(2\pi)^{n}} \, \chi_{T}(\xi)\, D(-\xi)\, ,
\label{integral}
\end{equation}
where the integral converges in the weak topology, see for instance~\cite[Cor. 5.3.5]{holevo_2011}. For more on what this means, see below.

It can be shown that the characteristic function of a Gaussian state $\rho_G[V,s]$ takes the form~\cite[Eq. (4.48)]{serafini_2017}
\begin{equation}
    \chi_{\rho_G[V,s]}(\xi) = \Tr[\rho_G[V,s]\, D(\xi)] = e^{-\frac14 \xi^T \Omega^T V \Omega \xi + i s^T \Omega \xi}\, .
    \label{chi G}
\end{equation}
Up to a change of variables, Eq.~\eqref{integral} can then be rewritten as follows:
\begin{equation}
    \rho_G[V,s] = \int \frac{d^{2n}\xi}{(2\pi)^n}\, e^{-\frac14 \xi^T V \xi - i s^T\xi} D(\Omega \xi)\, ,
    \label{integral G}
\end{equation}
where the integral converges weakly, see again~\cite[Cor. 5.3.5]{holevo_2011}. Among the other things, from Eq.~\eqref{integral} and~\eqref{chi G} it can be appreciated, that Gaussian states are exactly those quantum states whose characteristic function is a (multivariate) Gaussian.

A useful formula that we will employ in what follows gives the action of a random displacement on a Gaussian state: for all $K>0$, one has
\begin{equation}
    \int d^{2n}\xi\, \frac{e^{-\xi^T K^{-1} \xi}}{\pi^n \sqrt{\det K}}\ D(\xi)\, \rho_G [V,s]\, D(\xi)^\dag\, =\, \rho_G[V+K, s]\, ,
    \label{classical mixing}
\end{equation}
again in the sense of weak convergence. This can be seen as an immediate consequence of Eq.~\eqref{integral G}.

\subsection{Closedness and continuity results}

Let $\mathcal{H}_n$ be the Hilbert space associated with a finite number $n$ of harmonic oscillators, and let $\mathcal{T}(\mathcal{H}_n)$ be the set of trace class operators over $\mathcal{H}_n$. Observe that $\mathcal{T}(\mathcal{H}_n)$ becomes a Banach space once it is equipped with the trace norm $\|\cdot\|_1$, and its Banach dual is well known to be identifiable with the set of bounded operators, denoted by $\mathcal{B}(\mathcal{H}_n)$. Let us stress here that this is by no means a mathematical concept only. On the contrary, in quantum mechanics $\mathcal{B}(\mathcal{H}_n)$ has a physical interpretation as the set of all observables on the system.

In general, given a Banach space $E$ it is always possible to consider its (Banach) \emph{dual}, i.e., the space $E^*$ of all continuous linear functionals $\varphi: E\rightarrow \mathds{C}$.
Remember that a linear functional is continuous if and only if it is bounded, i.e., if and only if $\sup_{x\in E,\, \|x\|\leq 1} |\varphi(x)|$ is finite. The Banach dual can be used to induce another topology which is of interest, i.e., the \emph{weak topology}, defined as the coarsest topology that makes all the functionals in $E^*$ continuous. As a matter of fact, the topologies on $E$ such that the corresponding continuous dual is $E^*$ are exactly those that are coarser than the norm topology (induced by the norm on $E$) and finer than the weak topology. This is a special case of the Mackey-Arens theorem~\cite[Thm.~8.14]{aliprantis_2007}.
For a discussion of these concepts, see~\cite[Sec.~2.5]{megginson_1998} or~\cite[Sec.~3.11]{rudin_1991}.

If $E$ is infinite-dimensional it can be shown that the weak topology is always different (in fact, as the name suggests, strictly coarser) than the norm topology. Hence, when it comes to taking closures (something we shall be concerned with) one has to specify which topology is used, as in general the weak closure will be larger than the norm closure. However, this is not always the case. Indeed, there is an important class of sets for which weak and norm closure always coincide, i.e., that of convex sets (see~\cite[Thm.~2.5.16]{megginson_1998} or~\cite[Sec.~3.12]{rudin_1991}). By the above discussion, it should be clear by now that all topologies on a Banach space $E$ such that the corresponding continuous dual coincides with the Banach dual $E^*$ have in fact the same closed convex sets.

The Banach space we care about here is $\mathcal{T}(\mathcal{H}_n)$, hence the norm topology is induced by the trace norm $\|\cdot\|_1$, and the weak topology is nothing but the the coarsest topology that makes all linear functionals $\Tr [A (\cdot)] : \mathcal{T}(\mathcal{H}_n)\rightarrow \mathds{C}$ continuous, for all $A\in \mathcal{B}(\mathcal{H}_n)$.
Inside $\mathcal{T}(\mathcal{H}_n)$ lies the set of Gaussian states, denoted by $\mathcal{G}_n$, where $n$ is the number of modes. It is not completely trivial to show that $\mathcal{G}_n$ is norm-closed, and so we first show this result below.

\begin{lemma} \label{lemma G closed}
The set of Gaussian states $\mathcal{G}_n\subset \mathcal{T}(\mathcal{H}_n)$ is closed with respect to the topology induced by the trace norm.
\end{lemma}

\begin{proof}
We have to show that given a sequence $\rho^{(k)}_G$ of Gaussian states with the property that $\lim_k \|\rho^{(k)}_G - \rho\|_1$ for some trace class operator $\rho$, we have that $\rho$ itself is a Gaussian state. In what follows, we denote by $V(k)$ and $s(k)$ the covariance matrix and displacement vector of $\rho_G^{(k)}$, respectively, so that $\rho_G^{(k)}=\rho_G[V(k), s(k)]$.

The first step in the proof consists in showing that $V(k)$ and $s(k)$ are bounded sequences, i.e., that there exists $M\in\mathds{R}$ such that $\|V(k)\|_\infty, |s(k)|_2 \leq M$ for all $k$ (where $\|\cdot\|_\infty$ is the operator norm, and $|\cdot|_2$ the Euclidean norm for vectors). In order to see why, write
\begin{align*}
    &\braket{u| \rho_G[V(k), s(k)] |u} \\
    &\qquad = \int \frac{d^{2n}\xi}{(2\pi)^n}\, e^{-\frac14 \xi^T V(k) \xi - i s(k)^T\xi} \braket{u|D(\Omega \xi)|u} \\
    &\qquad\texteq{(1)} \int \frac{d^{2n}\xi}{(2\pi)^n}\, e^{-\frac14 \xi^T V(k) \xi - i s(k)^T\xi} \braket{0| D(-u) D(\Omega \xi) D(u)|0} \\
    &\qquad\texteq{(2)} \int \frac{d^{2n}\xi}{(2\pi)^n}\, e^{-\frac14 \xi^T V(k) \xi - i s(k)^T\xi} e^{-i u^T \xi} \braket{0| D(\Omega \xi) |0} \\
    &\qquad\texteq{(3)} \int \frac{d^{2n}\xi}{(2\pi)^n}\, e^{-\frac14 \xi^T V(k) \xi - i s(k)^T\xi} e^{-i u^T \xi} e^{-\frac14 \xi^T \xi} \\
    &\qquad= \int \frac{d^{2n}\xi}{(2\pi)^n}\, e^{-\frac14 \xi^T \left( V(k) + \id\right) \xi - i \left(s(k) + u\right)^T \xi} \\
    &\qquad\texteq{(4)} \frac{2^n e^{- (s(k)+u)^T (V(k)+\id)^{-1} (s(k)+u)}}{\sqrt{\det\left( V(k)+\id\right)}}\, .
\end{align*}
The justification of the above steps is as follows: (1) we used the definition of coherent states, Eq.~\eqref{coherent}; (2) we applied Eq.~\eqref{Weyl} twice; (3) we made use of Eq.~\eqref{D on 0}; (4) we performed the Gaussian integral. Now, we take the limit $k\rightarrow \infty$ on both sides of the equality
\begin{equation}
    \braket{u| \rho_G[V(k),s(k)] |u} = \frac{2^n e^{- (s(k)+u)^T (V(k)+\id)^{-1} (s(k)+u)}}{\sqrt{\det\left( V(k)+\id\right)}}\, .
    \label{G on coherent}
\end{equation}

On the left-hand side we get $\braket{u|\rho|u}$ because of the properties of the convergence in norm. Let us now look at the right-hand side. Observe that $\det\left( V(k)+\id\right)\geq \|V(k)\|_\infty+1$, and that the exponential term is at most $1$. If the sequence $V(k)$ were unbounded, then there would exist a subsequence $k_m$ on which $\|V(k_m)\|_\infty\rightarrow \infty$, which implies by the above equality that $\braket{u|\rho|u}=0$. Since this would happen for all $u\in\mathds{R}^{2n}$, we would deduce that $\braket{\psi|\rho|\psi}=0$ for all vectors $\ket{\psi}\in\mathcal{H}_n$, because coherent states are dense, and $\rho$ is a bounded (even trace class) operator. It is elementary to verify that this would imply that $\rho=0$ identically, a contradiction. Hence, we are led to conclude that $V(k)$ must be bounded, i.e., $V(k)\leq M\id$ for some $M\in\mathds{R}$.

This implies immediately that $(V(k)+\id)^{-1}\geq (M+1)^{-1} \id$, hence if the sequence $s(k)$ were unbounded, for every fixed $u$ we could find a subsequence $k_m$ on which $|s(k_m)+u|_2\rightarrow\infty$, which implies that
\begin{equation*}
    e^{-(s(k_m)+u)^T (V(k_m) + \id)^{-1} (s(k_m)+u)}\leq e^{-\frac{1}{M+1} |s(k_m)+u|_2^2}\underset{m\rightarrow\infty}{\longrightarrow} 0\, .
\end{equation*}
Since the determinant appearing in Eq.~\eqref{G on coherent} is always at least $1$, we would deduce that the whole r.h.s.\ of Eq.~\eqref{G on coherent} tends to $0$ on that subsequence, hence that $\braket{u|\rho|u}=0$ for all $u\in\mathds{R}^{2n}$, again a contradiction.

This shows that $V(k)$ and $s(k)$ form bounded sequences. Since they live in finite-dimensional spaces, they will admit two simultaneously convergent subsequences
\begin{align*}
    V(k_m) &\underset{m\rightarrow\infty}{\longrightarrow} V\, , \\
    s(k_m) &\underset{m\rightarrow\infty}{\longrightarrow} s\, .
\end{align*}
Clearly, one still has $\lim_{m\rightarrow\infty} \|\rho_G[V(k_m), s(k_m)] - \rho\|_1=0$. Now, we use this to take the limit $m\rightarrow \infty$ on both sides of the equality
\begin{equation}
    \Tr[\rho_G[V(k_m), s(k_m)]\, D(\xi) ] = e^{-\frac14 \xi^T \Omega^T V(k_m) \Omega \xi + i s(k_m)^T \Omega \xi}\, ,
\end{equation}
which is just a rewriting of Eq.~\eqref{chi G} (here, $\xi\in\mathds{R}^{2n}$ is fixed).
On the left-hand side we have
\begin{equation*}
    \lim_{m\rightarrow \infty} \Tr[\rho_G[V(k_m), s(k_m)]\, D(\xi) ] = \Tr [\rho D(\xi)] = \chi_\rho(\xi)
\end{equation*}
because the convergence of the sequence of states is in trace norm, and $D(\xi)$ is a bounded (even unitary) operator.
On the right-hand side, by our hypotheses
\begin{equation*}
    \lim_{m\rightarrow\infty} e^{-\frac14 \xi^T \Omega^T V(k_m) \Omega \xi + i s(k_m)^T \Omega \xi} = e^{-\frac14 \xi^T \Omega^T V \Omega \xi + i s^T \Omega \xi}\, .
\end{equation*}
The equality above then implies that
\begin{equation*}
    \chi_\rho(\xi) = e^{-\frac14 \xi^T \Omega^T V \Omega \xi + i s^T \Omega \xi}\, ,
\end{equation*}
from which we see that the limit state $\rho$ has a Gaussian characteristic function, hence it is Gaussian.
\end{proof}

There is another continuity result that we shall need in what follows. In a way, this can be considered as a strengthening of~\cite[Lemma 1]{depalma_2015}.

\begin{lemma} \label{lemma V to rho 1-continuous}
Consider a continuous variable system with $n$ degrees of freedom. The map
\begin{equation}
\begin{array}{rcl}
\text{\emph{QCM}}_n \oplus \mathds{R}^{2n} & \longrightarrow & \mathcal{T}(\mathcal{H}_n) \\[1ex]
(V,s) & \longmapsto & \rho_G[V,s]\, ,
\end{array}
\label{V to rho}
\end{equation}
which sends a pair $(V,s)$, where $V$ is a QCM and $s$ a real vector, is continuous with respect to the trace norm. Here, the topology on
\begin{equation*}
\text{\emph{QCM}}_n \oplus \mathds{R}^{2n} \subset \mathcal{M}_{2n}(\mathds{R}) \oplus \mathds{R}^{2n} \simeq \mathds{R}^{(2n)^2+2n}
\end{equation*}
is understood to be the standard one.
\end{lemma}

\begin{proof}
We have to show that whenever $\lim_{k\rightarrow \infty} V(k) = V$ and $\lim_{k\rightarrow \infty} s(k)=s$ one has also $\lim_{k\rightarrow\infty} \|\rho_G[V(k),s(k)]-\rho_G[V,s]\|_1=0$. At first glance we seem to have a problem here, as the trace distance of two Gaussian states is not a handy object when dealt with from the phase space perspective. However, we can exploit the Fuchs-van de Graaf's inequality $\|\rho-\sigma\|_1\leq 2\sqrt{1-F(\rho,\sigma)^2}$ to upper bound the trace distance by means of a fidelity-based quantity. The fidelity between two Gaussian states happens to have an explicit expression in terms of their first and second moments~\cite[Eq. (9)-(14)]{banchi_2015}. One can verify by direct inspection that this is continuous with respect to the involved covariance matrices and displacement vectors, and of course it reduces to $1$ when the first and second moments of the first state coincide with those of the second state. Hence,
\begin{align*}
    &\lim_{k\rightarrow\infty} \|\rho_G[V(k),s(k)]-\rho_G[V,s]\|_1 \\
    &\qquad \leq \lim_{k\rightarrow\infty} 2\sqrt{1-F(\rho_G[V(k), s(k)],\rho_G[V,s])^2} \\
    &\qquad= 2\sqrt{1- \left( \lim_{k\rightarrow\infty} F(\rho_G[V(k), s(k)],\rho_G[V,s]) \right)^2} \\
    &\qquad= 2\sqrt{1- \left( F(\rho_G[V, s],\rho_G[V,s]) \right)^2} \\
    &\qquad= 2\sqrt{1- \left( 1 \right)^2} \\
    &\qquad= 0\, ,
\end{align*}
as claimed.
\end{proof}


\section{Gaussian resources} \label{app G resources}

\subsection{Free states}

\begin{lemma} \label{lemma F closed}
Let $\tau$ be a linear topology on $\mathcal{T}(\mathcal{H}_N)$ (the space of trace-class operators) such that the corresponding continuous dual is $\left(\tau, \mathcal{T}(\mathcal{H}_N)\right)' =\mathcal{B}(\mathcal{H}_N)$ (the space of bounded operators). If Postulates~\ref{post convexity} and~\ref{post weak closed} hold, then all sets of free states $\mathcal{F}(\bs{\lambda})$ are closed with respect to $\tau$.
\end{lemma}

\begin{proof}
Since the weak topology on $\mathcal{T}(\mathcal{H}_N)$ is by definition the coarsest topology that makes all functionals $\Tr[A(\cdot)]:\mathcal{T}(\mathcal{H}_N)\rightarrow\mathds{C}$ continuous (where $A\in\mathcal{B}(\mathcal{H}_N)$ is generic), any topology $\tau$ that satisfies the hypothesis will be finer than the weak topology. Thus, it suffices to show that all sets $\mathcal{F}(\bs{\lambda})$ are weakly closed.
This follows since $\mathcal{F}(\bs{\lambda})$ are norm-closed and convex by assumption, and weak closure and norm closure always coincide for convex sets by Mazur's theorem (see e.g.~\cite[Thm.~2.5.16]{megginson_1998} or~\cite[Sec.~3.12]{rudin_1991}).
\end{proof}

\begin{lemma} \label{lemma FG closed}
When Postulate~\ref{post weak closed} holds, the set of Gaussian free states $\mathcal{F}_G(\bs{\lambda})$ is norm-closed.
\end{lemma}

\begin{proof}
By definition $\mathcal{F}_G(\bs{\lambda}) = \mathcal{F}(\bs{\lambda}) \cap \mathcal{G}_N$. The set $\mathcal{F}(\bs{\lambda})$ is norm-closed by Postulate~\ref{post weak closed}, and the set $\mathcal{G}_N$ of all Gaussian states is also norm-closed by Lemma~\ref{lemma G closed}. Since the intersection of closed sets is closed, we conclude.
\end{proof}

\begingroup
\renewcommand\theprop{\ref{prop V}}
\begin{prop}
If Postulate~\ref{post weak closed} holds, then the set $\mathcal{V}(\bs{\lambda})$ is topologically closed. If also Postulate~\ref{post convexity} holds, then $\mathcal{V}(\bs{\lambda})$ becomes `upward closed', in the sense that $V\in \mathcal{V}(\bs{\lambda})$ and $W\geq V$ implies $W\in\mathcal{V}(\bs{\lambda})$.
\end{prop}

\begin{proof}
We first show that $\mathcal{V}(\bs{\lambda})$ is topologically closed. By Lemma~\ref{lemma V to rho 1-continuous}, we know that the map $\Gamma: \text{QCM}_N\rightarrow \mathcal{T}(\mathcal{H}_N)$ whose action is defined by $\Gamma(V) \coloneqq \rho_G[V,0]$ is continuous with respect to the trace norm. With this notation, the set $\mathcal{V}(\bs{\lambda})$ can be rewritten as
\begin{equation*}
    \mathcal{V}(\bs{\lambda}) = \Gamma^{-1} \left( \mathcal{F}_G(\bs{\lambda}) \right)\, .
\end{equation*}
Since $\mathcal{F}_G(\bs{\lambda})$ is norm-closed by Lemma~\ref{lemma FG closed}, and the preimages of closed sets via continuous maps are closed, we conclude that $\mathcal{V}(\bs{\lambda})$ is closed as well.

We now move on to the second claim. Since we already showed that $\mathcal{V}(\bs{\lambda})$ is topologically closed, it is enough to show that is \emph{strictly} upward closed, i.e., that for all $V\in\mathcal{V}(\bs{\lambda})$ and $W>V$ one has also $W\in\mathcal{V}(\bs{\lambda})$. This is an easy consequence of Eq.~\eqref{classical mixing}. If we substitute there $K=W-V$, on the left-hand side we get a state in $\cl \left(\co \mathcal{F}_G(\bs{\lambda})\right)$, the closed convex hull of the set of free Gaussian states. From the right-hand side we learn that this state is actually a Gaussian state, hence by Gaussian convexity of the set $\mathcal{F}_G(\bs{\lambda})$ it must be also free. Finally, its covariance matrix is $V+K=W$, which leads us to conclude that $W\in\mathcal{V}(\bs{\lambda})$.
\end{proof}
\addtocounter{prop}{-1}
\endgroup

\begin{prop} \label{prop convexity}
Assume that Postulates~\ref{post tensor},~\ref{post partial trace} and~\ref{post weak closed} hold. Moreover, let the set of Gaussian free states be invariant under local mixing with 50:50 beam splitters, i.e., assume that for any pair of states $\rho,\sigma\in \mathcal{F}_G(\bs{\lambda})$ of a system with total number of modes $N$ one has
\begin{equation}
\left( \bigotimes_{j=1}^N U(\pi/4)_{j,j} \right) (\rho\otimes \sigma) \left( \bigotimes_{j=1}^N U(\pi/4)_{j,j} \right)^\dag \in \mathcal{F}_G(\bs{\lambda})\, ,
\end{equation}
where $U(\pi/4)_{j,j}$ is the unitary that implements the action of a 50:50 beam splitter on the $j$-th mode of $\rho$ and the same mode of $\sigma$. Then the corresponding set of free covariance matrices $\mathcal{V}(\bs{\lambda})$ is convex.
\end{prop}

\begin{proof}
Since $\mathcal{V}(\bs{\lambda})$ is topologically closed by Prop.~\ref{prop V}, it is convex if and only if it is midpoint convex, meaning that $\frac12 (V+W)\in\mathcal{V}(\bs{\lambda})$ whenever $V,W\in\mathcal{V}(\bs{\lambda})$.
Hence, let us show that $\mathcal{V}(\bs{\lambda})$ is midpoint convex. Picking $V,W$ as above, construct the state $\rho_G[V,0]\otimes \rho_G[W,0]$, which is free by Postulate~\ref{post tensor}, and whose covariance matrix is $V\oplus W = \left( \begin{smallmatrix} V & \\ & W \end{smallmatrix}\right)$. By hypothesis, mode-by-mode mixing with a 50:50 beam splitter yields another free state, whose covariance matrix will be
\begin{equation*}
   \frac12 \begin{pmatrix} V+W & V-W \\ V-W & V +W \end{pmatrix} .
\end{equation*}
Tracing away one of the two spatially separated subsystems leaves the other in a state with covariance matrix $\frac12 (V+W)$. Such a state must be free by Postulate~\ref{post partial trace}, hence we conclude that $\frac12 (V+W)\in\mathcal{V}(\bs{\lambda})$, as claimed.
\end{proof}

\subsection{Quantification and distillation}

We remind the reader that in general a Gaussian completely positive map $\Lambda$ from $A$ to $B$ acts on covariance matrices as follows~\cite{giedke_2002,fiurasek_2002}:
\begin{equation}
\Lambda: V_{A}\longmapsto \big( \Gamma_{AB} + \Sigma V_A \Sigma \big) \big/ (\Gamma_A + \Sigma V_A \Sigma)\, .
\label{G map}
\end{equation}
Here, $\Gamma_{AB}$ is the quantum covariance matrix associated with the Choi state of the map, and $\Sigma$ is the matrix that reverses the signs of all the momenta of the system on which it is acting, i.e.,
\begin{equation}
\Sigma \coloneqq \begin{pmatrix} \id & \\ & -\id \end{pmatrix}
\label{Sigma}
\end{equation}
according to the block decomposition of Eq.~\eqref{Omega}. The Schur complement of a $2\times 2$ block matrix $M = \left( \begin{smallmatrix} P & X \\ Y & Q \end{smallmatrix} \right)$ with respect to one of its square invertible blocks is given by
\begin{equation}
M/P \coloneqq Q - YP^{-1} X\, .
\label{Schur}
\end{equation}
It is elementary to verify that the above quantity behaves well under scalar multiplication, in the sense that $(\lambda M) / (\lambda P) = \lambda (M/P)$ for all scalars $\lambda\neq 0$. Furthermore, it is known that the Schur complement admits the following variational representation:
\begin{equation}
M/P = \max \big\{ R:\ M \geq 0 \oplus R \big\} \, ,
\label{Schur var}
\end{equation}
the ordering of the set on the right-hand side being the positive semidefinite (aka L\"owner) ordering. From Eq.~\eqref{Schur var} it follows in particular that $M/P$ is monotonically non-decreasing in $M>0$. For more details on the properties of Schur complements we refer the reader to the excellent monograph~\cite{zhang_2005}. A straightforward consequence of the above discussion is the following result.

\begin{lemma} \label{Gamma free}
If $\Gamma_{AB}$ represents a Gaussian free operation $\Lambda\in \mathcal{O}_G(\bs{\lambda}_A\rightarrow\bs{\lambda}_B)$, then
\begin{equation}
\big( \Gamma_{AB} + \Sigma V_A \Sigma \big) \big/ (\Gamma_A + \Sigma V_A \Sigma)\, \in\, \mathcal{V}(\bs{\lambda}_B)\qquad \forall\ V_A\in\mathcal{V}(\bs{\lambda}_A)\, .
\label{Gamma free eq1}
\end{equation}
Equivalently,
\begin{equation}
\forall\ V_A\in\mathcal{V}(\bs{\lambda}_A)\ \ \exists\ W_B\in\mathcal{V}(\bs{\lambda}_B) : \quad \Gamma_{AB} \geq (- \Sigma V_A \Sigma) \oplus W_B\, .
\label{Gamma free eq2}
\end{equation}
\end{lemma}

\begin{proof}
The first claim is a direct reformulation of the definition of resource non-generating operations, obtained via the explicit action of a Gaussian completely positive map as given by Eq.~\eqref{G map}. As for the second, let us observe that the inequality $\Gamma_{AB} \geq (- \Sigma V_A \Sigma) \oplus W_B$ implies, by Eq.~\eqref{Schur var}, that $\big( \Gamma_{AB} + \Sigma V_A \Sigma \big) \big/ (\Gamma_A + \Sigma V_A \Sigma)\geq W_B$. Since the right-hand side is the covariance matrix of a free state by hypothesis, and Prop.~\ref{prop V} holds, we deduce that the left-hand side is a free covariance matrix as well. The converse inequality is proved similarly, by realizing that Eq.~\eqref{Schur var} implies that
\begin{equation*}
\Gamma_{AB} + \Sigma V_A \Sigma \geq 0_A \oplus \big( \Gamma_{AB} + \Sigma V_A \Sigma \big) \big/ (\Gamma_A + \Sigma V_A \Sigma)\eqqcolon 0_A \oplus W_B\, ,
\end{equation*}
which leads immediately to $\Gamma_{AB}\geq (-\Sigma V_A \Sigma) \oplus W_B$, as claimed.
\end{proof}

We now come to the discussion of the properties of the $\kappa_{\F}$ function defined in Eq.~\eqref{kappa}. We first state some elementary facts.

\begin{lemma} \label{elementary properties kappa}
The set $T_{\F}(V) \coloneqq \lset t \geq 1 \bar t V\in \mathcal{V}(\bs{\lambda}) \rset$ is non-empty and topologically closed for all $V\in\text{QCM}_{N}$ as long as $\mathcal{V}(\bs{\lambda})$ is non-empty.
\end{lemma}

\begin{proof}
We first show that $T_{\F}(V)\neq \emptyset$ for all $V\in\text{QCM}_{N}$. Picking $W\in \mathcal{V}(\bs{\lambda})\neq \emptyset$, it is easy to see that
\begin{equation}
\|W\|_\infty \|V^{-1}\|_\infty\, V \geq W\, ,
\label{upper bound eigenvalues}
\end{equation}
where $\|\cdot\|_\infty$ denotes the operator norm. Observe that quantum covariance matrices are always strictly positive, hence $V^{-1}$ exists. We then write
\begin{align*}
\|W\|_\infty \|V^{-1}\|_\infty\, V &\geq \|W\|_\infty \|V^{-1}\|_\infty\, \lambda_{\min}(V) \id \\
&= \|W\|_\infty \|V^{-1}\|_\infty\, \|V^{-1}\|_\infty^{-1} \id \\
&= \|W\|_\infty \id \\
&\geq W\, .
\end{align*}
By the upward closedness of $\V_\F(\bs{\lambda})$ (Prop.~\ref{prop V}), one deduces immediately that $\max\left\{1, \|W\|_\infty \|V^{-1}\|_\infty \right\} \in T_{\F}(V)$, showing that the set is non-empty. To show that it is also topologically closed, just observe that
\begin{equation}
T_{\F}(V) = \Big([1,\infty)\cdot V\Big)\cap \mathcal{V}(\bs{\lambda})\, .
\end{equation}
The left-hand side of the above identity is the intersection of two closed sets, thanks to Prop.~\ref{prop V}, hence it is itself closed.
\end{proof}

The following is a refinement of Prop.~\ref{prop:kappa} from the main text.

\begin{prop} \label{prop:kappa-reformul}
The function $\kappa_{\F}(\cdot)$ defined by Eq.~\eqref{kappa} is:
\begin{enumerate}[(a)]
\item finite and well-defined for all $V\in\text{QCM}_N$;
\item faithful, in the sense that $\kappa_{\F}(V)=1$ if and only if $V\in\mathcal{V}(\bs{\lambda})$;
\item such that $\kappa_{\F}(s V)\geq s^{-1}\kappa_{\F}(V)$ for all $s\geq 1$;
\item monotonically non-increasing under $\OOG$; and
\item continuous.
\end{enumerate}
\end{prop}

\begin{proof}
Claim (a) follows directly from Lemma~\ref{elementary properties kappa}, while (b) is obvious from the definition. As for (c), one can distinguish two cases: if $s V\in \mathcal{V}(\bs{\lambda})$, then on the one hand $s \geq \kappa_{\F}(V)$, while on the other hand $\kappa_{\F}(s V)=1\geq s^{-1}\kappa_{\F}(V)$; if $s V\notin \mathcal{V}(\bs{\lambda})$, then
\begin{align*}
\kappa_{\F}(s V) &= \min \lset t\geq 1 \bar t s V\in  \mathcal{V}(\bs{\lambda}) \rset \\
&= s^{-1} \min\lset t'\geq s \bar t' V\in  \mathcal{V}(\bs{\lambda}) \rset \\
&= s^{-1} \min\lset t'\geq 1\bar t' V\in  \mathcal{V}(\bs{\lambda}) \rset \\
&= s^{-1} \kappa_{\F}(V)\, .
\end{align*}

We now turn to the proof of (d), i.e., the monotonicity of $\kappa_\F$ under Gaussian free operations. Call $\xi\coloneqq \kappa_\F (V)$. Then, by virtue of Eq.~\eqref{G map}, all we have to show is that $\kappa_\F \left( ( \Gamma_{AB} + \Sigma V_A \Sigma ) \big/ (\Gamma_A + \Sigma V_A \Sigma) \right)\leq \xi$, for all free Gaussian operations represented by covariance matrices $\Gamma_{AB}$ as in Lemma~\ref{Gamma free}. This amounts to proving that $\xi \big( \Gamma_{AB} + \Sigma V_A \Sigma \big) \big/ (\Gamma_A + \Sigma V_A \Sigma) \in \mathcal{V}(\bs{\lambda}_B)$. We write
\begin{align*}
&\xi \big( \Gamma_{AB} + \Sigma V_A \Sigma \big) \big/ (\Gamma_A + \Sigma V_A \Sigma) \\
&\qquad \texteq{(1)} \big( \xi \Gamma_{AB} + \xi \Sigma V_A \Sigma \big) \big/ (\xi \Gamma_A + \xi \Sigma V_A \Sigma) \\
&\qquad\textgeq{(2)} \big( \Gamma_{AB} + \xi \Sigma V_A \Sigma \big) \big/ ( \Gamma_A + \xi \Sigma V_A \Sigma) \\
&\qquad\overset{(3)}{\in} \mathcal{V}(\bs{\lambda}_B)\, .
\end{align*}
The justification of the above steps is as follows: (1) comes from homogeneity; (2) uses the monotonicity of the Schur complement, together with the observation that since $\xi=\kappa_\F (V)\geq 1$ one has $\xi \Gamma_{AB}\geq \Gamma_{AB}$; (3) is an elementary consequence of Eq.~\eqref{Gamma free eq1} applied to the free covariance matrix $\xi V$.

A useful observation that follows from the just established property (d) is that $\kappa_\F(\cdot)$ is also monotonically non-increasing with respect to the positive semidefinite ordering. In fact, adding some positive semidefinite matrix to the input never creates a resource state out of a free state, i.e., it is always a free operation.

What is left to show is claim (e). We will break the proof into two steps: first, we will show that $\limsup_{\Delta \rightarrow 0} \kappa_{\F} (V+\Delta) \leq \kappa_{\F}(V)$ for all $V>0$ (upper semicontinuity); second, we will complement this bound by means of the inequality $\liminf_{\Delta\rightarrow 0} \kappa_{\F} (V+\Delta) \geq \kappa_{\F}(V)$ (lower semicontinuity). Clearly, the two statements together imply that $\lim_{\Delta\rightarrow 0} \kappa_{\F} (V+\Delta) = \kappa_{\F}(V)$, which is claim (e).
Now, the upper semicontinuity rests on the upward closedness of $\mathcal{V}(\bs{\lambda})$. For a sufficiently small perturbation $\Delta$, write
\begin{align*}
V+\Delta &\geq V - \|\Delta\|_{\infty} \id \\
&\geq V - \|\Delta\|_{\infty} \|V^{-1}\|_{\infty} V \\
&= \left(1-  \|\Delta\|_{\infty} \|V^{-1}\|_{\infty}\right) V\, ,
\end{align*}
we deduce that
\begin{equation*}
\frac{V+\Delta}{1-\|\Delta\|_{\infty} \|V^{-1}\|_{\infty}} \geq V\, .
\end{equation*}
Using properties (c) and (d), this in turn implies that
\begin{align*}
\left( 1-\|\Delta\|_{\infty} \|V^{-1}\|_{\infty}\right) \kappa_{\F} (V+\Delta) &\leq \kappa_{\F} \left( \frac{V+\Delta}{1-\|\Delta\|_{\infty} \|V^{-1}\|_{\infty}} \right) \\
&\leq \kappa_{\F}(V)\, ,
\end{align*}
from which it follows that
\begin{equation*}
\kappa_{\F}(V+\Delta) \leq \frac{\kappa_{\F}(V)}{1-\|\Delta\|_\infty \|V^{-1}\|_\infty}\, .
\end{equation*}
In particular,
\begin{equation*}
\limsup_{\Delta\rightarrow 0} \kappa_{\F} (V+\Delta) \leq \lim_{\Delta\rightarrow 0} \frac{\kappa_{\F}(V)}{1-\|\Delta\|_{\infty} \|V^{-1}\|_{\infty}} = \kappa_\F (V) \, ,
\end{equation*}
which proves upper semicontinuity. The lower semicontinuity comes instead from the topological closedness of the set $\mathcal{V}(\bs{\lambda})$, as established by Prop.~\ref{prop V}. To see why, consider a $V>0$ and a sequence of sufficiently small perturbation matrices $(\Delta_{n})_{n\in\mathds{N}}$ such that $\lim_{n\rightarrow\infty} \Delta_{n} = 0$. Since $\kappa_{\F} (V+ \Delta_{n})\, (V+\Delta_n)\in \mathcal{V}(\bs{\lambda})$ for all $n$, taking a subsequence $(n_{k})_{k\in\mathds{N}}$ such that $\lim_{k\rightarrow \infty} \kappa_{\F} (V+\Delta_{n_{k}}) = \liminf_{n\rightarrow \infty} \kappa_{\F} (V+\Delta_{n})$, we obtain by closedness that
\begin{align*}
\mathcal{V}(\bs{\lambda}) &\ni \lim_{k\rightarrow\infty} \left( \kappa_{\F} (V+\Delta_{n_{k}})\, (V+\Delta_{n_{k}})\right) \\
&= \left( \lim_{k\rightarrow\infty} \kappa_\F (V+\Delta_{n_k}) \right) V \\
&= \left( \liminf_{n\rightarrow \infty} \kappa_{\F} (V+\Delta_{n}) \right) V\, ,
\end{align*}
which implies in turn that $\kappa_{\F} (V) \leq \liminf_{n\rightarrow \infty} \kappa_{\F} (V+\Delta_{n})$, proving lower semicontinuity and hence claim (e).
\end{proof}

\begin{rem}
In fact, we have shown that
\begin{equation}
\kappa_{\F} (V) \leq \max\left\{ 1,\,\|V^{-1}\|_\infty \min_{W\in \mathcal{V}(\bs{\lambda})} \|W\|_\infty\right\} .
\end{equation}
for all $V\in\text{QCM}_{N}$.
\end{rem}

\begin{rem}
An inspection of the above proof of the monotonicity result (Prop.~\ref{prop:kappa}(d)) reveals that the only property of the Choi covariance matrix $\Gamma_{AB}$ we have made use of is its positive semidefiniteness. This observation shows that $\kappa_{\F}$ is monotonic under any operation of the form specified by Eq.~\eqref{G map} with $\Gamma_{AB}\geq 0$.
\end{rem}

The fundamental property of the $\kappa_\F$ measure we employ here concerns its behaviour when multiple copies of the same state are considered. 

\begin{lemma} \label{kappa multi-copy}
For all $\bs{\lambda},\bs{\mu}$, consider the $\kappa_\F$ functions identified via Eq.~\eqref{kappa} by the sets of free covariance matrices $\mathcal{V}(\bs{\lambda})$, $\mathcal{V}(\bs{\mu})$, and $\mathcal{V}(\bs{\lambda}\oplus\bs{\mu})$. Then for all $V\in\text{\emph{QCM}}_N$ and $W\in\text{\emph{QCM}}_M$, where $N=\sum_j n_j$ and $M=\sum_j m_j$, it holds that
\begin{equation*}
\kappa_\F \left(V \oplus W\right) = \max\{ \kappa_\F (V), \kappa_\F (W)\}\, .
\end{equation*}
\end{lemma}

\begin{proof}
Call $\eta\coloneqq \max\{ \kappa_\F (V), \kappa_\F (W)\}$. From Prop.~\ref{prop V} and from the inequalities $\eta\geq \kappa_\F (V),\kappa_\F(W)$ we deduce immediately that $\eta V\in\mathcal{V}(\bs{\lambda})$, $\eta W\in \mathcal{V}(\bs{\mu})$. By Postulate~\ref{post tensor}, we deduce that
\begin{equation*}
\eta \left( V\oplus W \right) = (\eta V)\oplus (\eta W) \in \mathcal{V}(\bs{\lambda})\oplus \mathcal{V}(\bs{\mu})\subseteq \mathcal{V}(\bs{\lambda}\oplus \bs{\mu})\, ,
\end{equation*}
which implies by definition that $\kappa_\F (V\oplus W)\leq \eta = \max\{ \kappa_\F (V)_\F, \kappa_\F (W)\}$. As for the opposite inequality, call $\zeta \coloneqq \kappa_\F \left(V \oplus W\right)$. Then $\zeta (V\oplus W)\in\mathcal{V}(\bs{\lambda}\oplus \bs{\mu})$, and by Postulate~\ref{post partial trace} we can generate a free state of the first system by tracing away the second. At the level of covariance matrices this amounts to performing a local projection, for which we adopt the same notation as in the characterization of Postulate~\ref{post partial trace} in the manuscript. We then obtain
\begin{equation*}
\zeta V = \Pi \left( \zeta (V\oplus W) \right) \Pi^T \in \Pi\, \mathcal{V}(\bs{\lambda}\oplus\bs{\mu})\, \Pi^T \subseteq \mathcal{V}(\bs{\lambda})\, .
\end{equation*}
This shows that $\kappa_\F(V) \leq \zeta = \kappa_\F (V\oplus W)$. Repeating the reasoning with $W$ instead of $V$ we get also $\kappa_\F (W)\leq \kappa_\F (V\oplus W)$, and putting the two inequalities together we have $\max\{ \kappa_\F(V), \kappa_\F(W)\} \leq \kappa_\F (V\oplus W)$, which completes the proof.
\end{proof}

\begingroup
\renewcommand\thethm{\ref{thm:no-go}}
\begin{thm}
Consider an arbitrary Gaussian resource theory satisfying Postulates~\ref{post D invariance}-\ref{post weak closed} and two covariance matrices $V,W\in \text{\emph{QCM}}_N$. If $\kappa_\F (W) > \kappa_\F (V)$, then it is impossible to find a sequence $(W_{n})_{n\in\mathds{N}}\subset\text{\emph{QCM}}_{N}$ such that $\lim_{n\rightarrow\infty} W_{n}=W$ and the transformations $V^{\oplus n}\rightarrow W_{n}$ are possible with Gaussian resource non-generating operations for all $n$.
\end{thm}
\begin{proof}
If said transformation were possible, by combining Prop.~\ref{prop:kappa} and Lemma~\ref{kappa multi-copy} one would obtain
\begin{equation*}
\kappa_\F (V) = \kappa_\F\left( V^{\oplus n} \right) \geq \kappa_\F (W_n)\, .
\end{equation*}
Since $\kappa_\F$ is continuous, one has $\lim_{n\rightarrow\infty} \kappa_\F (W_n) = W$ and hence also $\kappa_\F(V)\geq \kappa_\F (W)$, which is a contradiction.
\end{proof}
\addtocounter{thm}{-1}
\endgroup

\begin{rem}
The remark after Prop.~\ref{prop:kappa} has an important consequence here. Namely, we now see that the above no-go result still holds if one allows as free operations all resource non-generating maps of the form given by Eq.~\eqref{G map} with $\Gamma_{AB}\geq 0$. Remember that a map acting on the second moments as in Eq.~\eqref{G map} is a valid physical transformation (completely positive map) if and only if $\Gamma_{AB}$ is a quantum covariance matrix, i.e., if and only if $\Gamma_{AB}\geq i\Omega_{AB}$. Since this is a strictly stronger constraint than simply requiring that $\Gamma_{AB}\geq 0$, this observation extends the validity of Thm.~\ref{thm:no-go} even further. For instance, its domain of applicability now includes the maps considered in~\cite[Eq.~(24)-(26)]{depalma_2015}, since the corresponding Choi covariance matrices can be shown to be positive semidefinite provided~\cite[Eq.~(27)]{depalma_2015} is obeyed. However, as some of these maps will be unphysical, the extension discussed here may be regarded mainly as a mathematical curiosity.
\end{rem}


\section{Semidefinite programming representation of Gaussian resources} \label{app SDP}

\subsection{Quantum entanglement}

Recall that the characterization of the set of separable states $\rho_G[V_{AB},s] \in \S_{A|B}$ can be simplified to~\cite{lami_2016}
\begin{equation}\label{eq:sep_simplified}
	\rho_G[V_{AB},s] \in \S_{A|B} \iff V_{AB} \geq \gamma_A \oplus i\Omega_B,
\end{equation}
which gives the following semidefinite representation of the quantifier $\kappa_\S$:
\begin{equation}
\begin{aligned}
\hspace{-5em}\kappa_\S(V_{AB}) = \,& \underset{\lambda, \gamma_{A}}{\text{minimize}}
& & \lambda\\
& \text{subject to} & & \lambda \, V_{AB} \geq \gamma_A \oplus i \Omega_B\\
&&& \gamma_A = \gamma_A^T\\
&&& \gamma_A \geq i\Omega_A\\
&&& \lambda \geq 1
\end{aligned}
\end{equation}
where one can equivalently consider the subsystem $B$ instead. The Lagrange dual of $\upsilon_\S$ can be obtained as
\begin{equation}
\begin{aligned}
\upsilon_\S(V_{AB}) = \,& \underset{W, X}{\text{minimize}}
& & \<W, V_{AB}\>\\
& \text{subject to} & &  \< W_{22}, i\Omega_B \> + \< X,  i\Omega_A \> = 1\\
&&& \Re(W_{11}) = \Re(X)\\
&&& W, X \geq 0\\
\end{aligned}
\end{equation}
where $W = \left(\begin{smallmatrix}W_{11} & W_{12}\\ W_{12}^\dagger & W_{22}\end{smallmatrix}\right)$ and we use the Hilbert-Schmidt inner product $\<X,Y\> = \Tr(XY)$. With respect to the dual problem in Ref.~\cite{hyllus_2006-1} which requires an optimization over the spaces of Hermitian matrices $\Herm_{2n} \oplus \Herm_{2n}$, using the simplified the condition for separability in Eq.~\eqref{eq:sep_simplified} reduces the optimization space to $\Herm_{2n}\oplus\Herm_{2 n_A}$.

To see that we were justified in claiming that the optimal value of $\upsilon_\S$ is equal to the optimal value of the dual, we will show that strong duality holds. Take $W^\star\oplus X^\star = \mathbbm{1}_{2n+2n_A} + \frac{i}{2n} \Omega\oplus\Omega_A$, and notice that $W^\star \oplus X^\star > 0$ since it is Hermitian and all of its eigenvalues are given by $\frac{2n\pm 1}{2n}>0$, and $\Tr\left(W_{22}^\star i \Omega_B \right) + \Tr\left(X i \Omega_A\right) = 1$. This means that $W^\star$ and $X^\star$ form a strictly feasible solution to the dual problem, and so Slater's condition is satisfied and strong duality holds~\cite{boyd_2004}.

\subsection{Steering}

The primal problem corresponding to the quantifier of $A\to B$ steerability $\kappa_\T$ is then given by
\begin{equation}
\begin{aligned}
\hspace{-5em}\kappa_\T(V_{AB}) = \,& \underset{\lambda\geq 1}{\text{minimize}}
& & \lambda\\
& \text{subject to} & & \lambda \, V_{AB} \geq  0_A \oplus i \Omega_B.
\end{aligned}
\end{equation}
An important property of the Schur complement is that, given a Hermitian matrix $M = \left( \begin{smallmatrix} P & X \\ X^\dagger & Q \end{smallmatrix} \right)$ such that $P > 0$, we have $M \geq 0 \iff M/P \geq 0$~\cite{zhang_2005}. Now, since $V_A > 0$, we can equivalently write
\begin{equation}
\begin{aligned}\label{eq:steering_dual}
\hspace{-5em}\kappa_\T(V_{AB}) = \,& \underset{\lambda \geq 1}{\text{minimize}}
& & \lambda\\
& \text{subject to} & & \lambda \, V_{AB} / V_A \geq i \Omega_B.
\end{aligned}
\end{equation}
The corresponding inverse dual is given as
\begin{equation}
\begin{aligned}
\hspace{-31em}{\upsilon_\T}(V_{AB}) = \,& \underset{W}{\text{minimize}}
& & \< W,  V_{AB} \> \\
& \text{subject to} & &  \< W_{22}, i\Omega_B \> = 1\\
&&& W \geq 0\\
\hphantom{\hspace{-31em}{\upsilon_\T}(V_{AB})} = \,& \underset{W}{\text{minimize}}
& & \< W,  V_{AB} / V_A \> \\
& \text{subject to} & &  \<  W,  i\Omega_B \> = 1\\
&&& W \geq 0.
\end{aligned}
\end{equation}
Taking $W^\star = \mathbbm{1}_{2n} + \frac{i}{2n} \Omega$, we have that $W^\star > 0$ since it is Hermitian and its eigenvalues are given by $\frac{2n\pm 1}{2n}>0$, and $\Tr\left(W^\star i \Omega\right) = 1$. This means that $W^\star$ is a strictly feasible solution to the latter dual problem, so Slater's condition is satisfied and strong duality holds.

In fact, Eq.~\ref{eq:steering_dual} suggests an interesting alternative characterization of this quantifier in terms of a symplectic eigenvalue problem. To see this, consider first the following result (see also~\cite{bhatia_2015}).

\begin{prop}
The smallest symplectic eigenvalue $\nu_{\min}(V)$ of any $V > 0$ can be expressed as
\begin{equation}\begin{aligned}\label{eq:smin}
	\nu_{\min}(V) &= \max \lset \lambda \geq 0 \bar V \geq i\lambda\Omega \rset\\
	&= \min \lset \< W, V \> \bar \< W, i\Omega \> = 1 \rset.
\end{aligned}\end{equation}
\end{prop}
\begin{proof}
Recall that a matrix $S$ is called symplectic if $S \Omega S^T = \Omega$. By Williamson's theorem~\cite{williamson_1936,simon_1999}, there exists a symplectic matrix $S$ such that $S V S^T = D \oplus D$ with $D=\diag\left(\nu_1(V),\ldots,\nu_n(V)\right) > 0$ being a diagonal matrix of the symplectic eigenvalues of $V$. We then have
\begin{equation}\begin{aligned}
\max \lset \lambda \bar V \geq i \lambda \Omega \rset &\texteq{(1)} \max \lset \lambda \bar SVS^T \geq i \lambda S \Omega S^T \rset\\
&= \max \lset \lambda \bar D \oplus D \geq i \lambda \Omega \rset\\
&\texteq{(2)} \max \lset \lambda \bar D - \lambda^2 D^{-1} \geq 0\rset\\
&\texteq{(3)} \max \lset \lambda \bar \nu_j(V)^2 - \lambda^2 \geq 0 \; \forall j \rset\\
&= \nu_{\min}(V)
\end{aligned}\end{equation}
where (1) follows since any symplectic matrix is non-singular, (2) follows from the Schur complement condition for positive semidefiniteness, and (3) follows since both $D$ and $D^{-1}$ are diagonal with $\nu_j(V) > 0 \;\forall j$. The second line of Eq.~\eqref{eq:smin} then follows by strong Lagrange duality.
\end{proof}

This leads to the following simple representation:
\begin{equation}\begin{aligned}
	\upsilon_\T (V_{AB}) = \nu_{\min}\left(V_{AB} / V_A\right).
\end{aligned}\end{equation}
The quantifier can thus be related to a commonly used measure, the Gaussian $A\to B$ steerability~\cite{kogias_2015} $N_\T(V_{AB}) \coloneqq -\sum_k \log \min\{1, \nu_k(V_{AB} / V_A) \}$. In particular, in the case of a bipartite system where $n_B = 1$, $V_{AB}/V_A$ has only one symplectic eigenvalue and therefore we have $N_\T(V_{AB}) = \log \kappa_\T(V_{AB})$.

\end{document}